\documentclass[times, 11pt]{article}
\usepackage{latex8}
\usepackage{times}
\usepackage{amsmath}
\usepackage{amssymb}
\newtheorem{theorem}{Theorem}

\newtheorem{definition}[theorem]{Definition}
\newtheorem{lemma}[theorem]{Lemma}
\newtheorem{proposition}[theorem]{Proposition}

\newtheorem{claim}[theorem]{Claim}

\newtheorem{remk}[theorem]{Remark}
\newtheorem{exmp}[theorem]{Example}

\def\FullBox{\hbox{\vrule width 8pt height 8pt depth 0pt}}

\def\qed{\ifmmode\qquad\FullBox\else{\unskip\nobreak\hfil
\penalty50\hskip1em\null\nobreak\hfil\FullBox
\parfillskip=0pt\finalhyphendemerits=0\endgraf}\fi}

\def\qedsketch{\ifmmode\Box\else{\unskip\nobreak\hfil
\penalty50\hskip1em\null\nobreak\hfil$\Box$
\parfillskip=0pt\finalhyphendemerits=0\endgraf}\fi}

\newenvironment{proof}{\begin{trivlist} \item {\bf Proof:~~}}
  {\qed\end{trivlist}}

\newcommand{\eqdef}{\mathbin{\stackrel{\rm def}{=}}}

\begin{document}

\title{Locally Decodable Codes From Nice Subsets of Finite Fields  \\
       and Prime Factors of Mersenne Numbers}
\author{Kiran S. Kedlaya \\ MIT \\ kedlaya@mit.edu \and Sergey Yekhanin \\ MIT \\
yekhanin@mit.edu} \maketitle \thispagestyle{empty}

\begin{abstract}
A $k$-query Locally Decodable Code (LDC) encodes an $n$-bit
message $x$ as an $N$-bit codeword $C(x),$ such that one can
probabilistically recover any bit $x_i$ of the message by querying
only $k$ bits of the codeword $C(x)$, even after some constant
fraction of codeword bits has been corrupted. The major goal of
LDC related research is to establish the optimal trade-off between
length and query complexity of such codes.

Recently~\cite{Y_nice} introduced a novel technique for
constructing locally decodable codes and vastly improved the upper
bounds for code length. The technique is based on Mersenne primes.
In this paper we extend the work of~\cite{Y_nice} and argue that
further progress via these methods is tied to progress on an old
number theory question regarding the size of the largest prime
factors of Mersenne numbers.

Specifically, we show that every Mersenne number $m = 2^t - 1$
that has a prime factor $p > m^{\gamma}$ yields a family of
$k(\gamma)$-query locally decodable codes of length
$\exp\left(n^{1/t}\right).$ Conversely, if for some fixed $k$ and
all $\epsilon>0$ one can use the technique of~\cite{Y_nice} to
obtain a family of $k$-query LDCs of length
$\exp\left(n^\epsilon\right);$ then infinitely many Mersenne
numbers have prime factors larger than known currently.
\end{abstract}

\section{Introduction}\label{Sec:Introduction}
Classical error-correcting codes allow one to encode an $n$-bit
string $x$ into in $N$-bit codeword $C(x),$ in such a way that $x$
can still be recovered even if $C(x)$ gets corrupted in a number
of coordinates. It is well-known that codewords $C(x)$ of length
$N=O(n)$ already suffice to correct errors in up to $\delta N$
locations of $C(x)$ for any constant $\delta<1/4.$ The
disadvantage of classical error-correction is that one needs to
consider all or most of the (corrupted) codeword to recover
anything about $x.$ Now suppose that one is only interested in
recovering one or a few bits of $x.$ In such case more efficient
schemes are possible. Such schemes are known as locally decodable
codes (LDCs). Locally decodable codes allow reconstruction of an
arbitrary bit $x_i,$ from looking only at $k$ randomly chosen
coordinates of $C(x),$ where $k$ can be as small as $2.$ Locally
decodable codes have numerous applications in complexity
theory~\cite{KT,T_survey}, cryptography~\cite{CGKS,Gasarch} and
the theory of fault tolerant computation~\cite{Romashchenko}.
Below is a slightly informal definition of LDCs:

A $(k,\delta,\epsilon)$-locally decodable code encodes $n$-bit
strings to $N$-bit codewords $C(x),$ such that for every $i\in
[n],$ the bit $x_i$ can be recovered with probability
$1-\epsilon,$ by a randomized decoding procedure that makes only
$k$ queries, even if the codeword $C(x)$ is corrupted in up to
$\delta N$ locations.

One should think of $\delta>0$ and $\epsilon<1/2$ as constants.
The main parameters of interest in LDCs are the length $N$ and the
query complexity $k.$ Ideally we would like to have both of them
as small as possible. The concept of locally decodable codes was
explicitly discussed in various papers in the early
1990s~\cite{BFLS,Sud,PS}. Katz and Trevisan~\cite{KT} were the
first to provide a formal definition of LDCs. Further work on
locally decodable codes
includes~\cite{BIK,DJKLR,Obata,BIKR,KdW,WdW,Y_nice,Woodruff,HO,Ragh}.

Below is a brief summary of what was known regarding the length of
LDCs prior to~\cite{Y_nice}. The length of optimal $2$-query LDCs
was settled by Kerenidis and de Wolf in~\cite{KdW} and is
$\exp(n).$\footnote{Throughout the paper we use the standard
notation $\exp(x)\eqdef e^{O(x)}.$} The best upper bound for the
length of $3$-query LDCs was $\exp\left(n^{1/2}\right)$ due to
Beimel et al.~\cite{BIK}, and the best lower bound is
$\tilde\Omega(n^2)$~\cite{Woodruff}. For general (constant) $k$
the best upper bound was $\exp\left(n^{O\left( \log\log k/(k\log
k)\right)}\right)$ due to Beimel et al.~\cite{BIKR} and the best
lower bound is $\tilde\Omega\left(n^{1+1/\left(\lceil k/2\rceil
-1\right)}\right)$~\cite{Woodruff}.

\smallskip

The recent work~\cite{Y_nice} improved the upper bounds to the
extent that it changed the common perception of what may be
achievable~\cite{Goldreich,Gasarch}.~\cite{Y_nice} introduced a
novel technique to construct codes from so-called nice subsets of
finite fields and showed that every Mersenne prime $p=2^t-1$
yields a family of $3$-query LDCs of length
$\exp\left(n^{1/t}\right).$ Based on the largest known Mersenne
prime~\cite{CB_Mersenne}, this translates to a length of less than
$\exp\left(n^{10^{-7}}\right).$ Combined with the recursive
construction from~\cite{BIKR}, this result yields vast
improvements for all values of $k>2.$ It has often been
conjectured that the number of Mersenne primes is infinite. If
indeed this conjecture holds,~\cite{Y_nice} gets three query
locally decodable codes of length
$N=\exp\left(n^{O\left(\frac{1}{\log\log n}\right)}\right)$ {\it
for infinitely many} $n.$ Finally, assuming that the conjecture of
Lenstra, Pomerance and
Wagstaff~\cite{LPW_conjecture,Pomerance,Wagstaff} regarding the
density of Mersenne primes holds,~\cite{Y_nice} gets three query
locally decodable codes of length
$N=\exp\left(n^{O\left(\frac{1}{\log^{1-\epsilon}\log
n}\right)}\right)$ {\it for all} $n,$ for every~$\epsilon
> 0.$

\subsection{Our results}
In this paper we address two natural questions left open
by~\cite{Y_nice}:
\begin{enumerate}

\item Are Mersenne primes necessary for the constructions
of~\cite{Y_nice}?

\item Has the technique of~\cite{Y_nice} been pushed to its
limits, or one can construct better codes through a more clever
choice of nice subsets of finite fields?

\end{enumerate}

We extend the work of~\cite{Y_nice} and answer both of the
questions above. In what follows let $P(m)$ denote the largest
prime factor of $m.$ We show that one does not necessarily need to
use Mersenne primes. It suffices to have Mersenne numbers with
polynomially large prime factors. Specifically, every Mersenne
number $m = 2^t - 1$ such that $P(m)\geq m^\gamma$ yields a family
of $k(\gamma)$-query locally decodable codes of length
$\exp\left(n^{1/t}\right).$ A partial converse also holds. Namely,
if for some fixed $k\geq 3$ and all $\epsilon>0$ one can use the
technique of~\cite{Y_nice} to (unconditionally) obtain a family of
$k$-query LDCs of length $\exp\left(n^\epsilon\right);$ then for
infinitely many $t$ we have
\begin{equation}
\label{Eqn:ImpliedPrimeFactor}
 P(2^t-1) \geq (t/2)^{1+1/(k-2)}.
\end{equation}

The bound~(\ref{Eqn:ImpliedPrimeFactor}) may seem quite weak in
light of the widely accepted conjecture saying that the number of
Mersenne primes is infinite. However (for any $k\geq 3$) this
bound is substantially stronger than what is currently known
unconditionally. Lower bounds for $P(2^t-1)$ have received a
considerable amount of attention in the number theory
literature~\cite{Schinzel,Stewart75,ErdosShorey,Stewart77,MurtyWong,MurataPomerance}.
The strongest result to date is due to Stewart~\cite{Stewart77}.
It says that for all integers $t$ ignoring a set of asymptotic
density zero, and for all functions $\epsilon(t)>0$ where
$\epsilon(t)$ tends to zero monotonically and arbitrarily slowly:
\begin{equation}
\label{Eqn:SteImpliedPrimeFactor}
 P(2^t-1) > \epsilon(t)t\left(\log t\right)^2/\log\log t.
\end{equation}
There are no better bounds known to hold for infinitely many
values of $t,$ unless one is willing to accept some number
theoretic conjectures~\cite{MurtyWong,MurataPomerance}. We hope
that our work will further stimulate the interest in proving lower
bounds for $P(2^t-1)$ in the number theory community.

\smallskip

In summary, we show that one may be able to improve the
unconditional bounds of~\cite{Y_nice} (say, by discovering a new
Mersenne number with a very large prime factor) using the same
technique. However any attempts to reach the
$\exp\left(n^{\epsilon}\right)$ length for some fixed query
complexity and all $\epsilon>0$ require either progress on an old
number theory problem or some radically new ideas.

\smallskip

In this paper we deal only with binary codes for the sake of
clarity of presentation. We remark however that our results as
well as the results of~\cite{Y_nice} can be easily generalized to
larger alphabets. Such generalization will be discussed in detail
in~\cite{Y_thesis}.

\subsection{Outline}
In section~\ref{Sec:SetsAndCodes} we introduce the key concepts
of~\cite{Y_nice}, namely that of combinatorial and algebraic
niceness of subsets of finite fields. We also briefly review the
construction of locally decodable codes from nice subsets. In
section~\ref{Sec:FactorsToSets} we show how Mersenne numbers with
large prime factors yield nice subsets of prime fields. In
section~\ref{Sec:SetsToFactors} we prove a partial converse.
Namely, we show that every finite field $\mathbb{F}_q$ containing
a sufficiently nice subset, is an extension of a prime field
$\mathbb{F}_p,$ where $p$ is a large prime factor of a large
Mersenne number. Our main results are summarized in
sections~\ref{SubSec:PositiveSummary}
and~\ref{SubSec:NegativeSummary}.

\section{Notation}
We use the following standard mathematical notation:
\begin{itemize}
\item $[s]=\{1,\ldots,s\};$

\item $\mathbb{Z}_n$ denotes integers modulo $n;$

\item $\mathbb{F}_q$ is a finite field of $q$ elements;

\item $d_H(x,y)$ denotes the Hamming distance between binary
vectors $x$ and $y;$

\item $(u,v)$ stands for the dot product of vectors $u$ and $v;$

\item For a linear space $L\subseteq \mathbb{F}_2^m,$ $L^\perp$
denotes the {\it dual} space. That is, $L^\perp=\{u\in
\mathbb{F}_2^m\ |\ \forall v\in L, (u,v)=0\};$

\item For an odd prime $p,$ $\mathrm{ord}_2(p)$ denotes the
smallest integer $t$ such that $p\ |\ 2^t-1.$

\end{itemize}

\section{Nice subsets of finite fields and locally decodable codes}\label{Sec:SetsAndCodes}
In this section we introduce the key technical concepts
of~\cite{Y_nice}, namely that of combinatorial and algebraic
niceness of subsets of finite fields. We briefly review the
construction of locally decodable codes from nice subsets. Our
review is concise although self-contained. We refer the reader
interested in a more detailed and intuitive treatment of the
construction to the original paper~\cite{Y_nice}. We start by
formally defining locally decodable codes.

\begin{definition}
A binary code $C:\{0,1\}^n\rightarrow \{0,1\}^N$ is said to be
$(k,\delta,\epsilon)$-locally decodable if there exists a
randomized decoding algorithm $\mathcal A$ such that
\begin{enumerate}
\item For all $x\in \{0,1\}^n,$  $i\in [n]$ and $y\in \{0,1\}^N$
such that $d_H(C(x),y)\leq \delta N:$ \mbox{Pr}$[{\mathcal
A}^y(i)=x_i]\geq 1-\epsilon,$ where the probability is taken over
the random coin tosses of the algorithm $\mathcal A.$

\item $\mathcal A$ makes at most $k$ queries to $y.$
\end{enumerate}
\end{definition}

We now introduce the concepts of combinatorial and algebraic
niceness of subsets of finite fields. Our definitions are
syntactically slightly different from the original definitions
in~\cite{Y_nice}. We prefer these formulations since they are more
appropriate for the purposes of the current paper. In what follows
let $\mathbb{F}_q^*$ denote the multiplicative group of
$\mathbb{F}_q.$

\begin{definition}
\label{Def:CombinatoriallyNice} A set $S\subseteq \mathbb{F}_q^*$
is called $t$ combinatorially nice if for some constant $c>0$ and
every positive integer $m$ there exist two $n=\lfloor
cm^t\rfloor$-sized collections of vectors $\{u_1,\ldots,u_n\}$ and
$\{v_1,\ldots,v_n\}$ in $\mathbb{F}_q^m,$ such that
\begin{itemize}
\item For all $i\in [n],$ $(u_i,v_i)=0;$

\item For all $i,j\in [n]$ such that $i\ne j,$ $(u_j,v_i)\in S.$
\end{itemize}
\end{definition}

\begin{definition}
\label{Def:AlgebraicallyNice} A set $S\subseteq \mathbb{F}_q^*$ is
called $k$ algebraically nice if $k$ is odd and there exists an
odd $k^\prime\leq k$ and two sets $S_0,S_1\subseteq \mathbb{F}_q$
such that
\begin{itemize}
\item $S_0$ is not empty;

\item $|S_1|=k^\prime;$

\item For all $\alpha\in \mathbb{F}_q$ and $\beta\in S:$
$\left|S_0\cap \left(\alpha+\beta S_1\right)\right|\equiv 0\mod
(2).$
\end{itemize}
\end{definition}

The following lemma shows that for an algebraically nice set $S,$
the set $S_0$ can always be chosen to be large. It is a
straightforward generalization of~\cite[lemma 15]{Y_nice}.

\begin{lemma}
\label{Lemma:S0Large} Let $S\subseteq \mathbb{F}_q^*$ be a $k$
algebraically nice set. Let $S_0,S_1\subseteq \mathbb{F}_q$ be
sets from the definition of algebraic niceness of $S.$ One can
always redefine the set $S_0$ to satisfy $|S_0|\geq \lceil q/2
\rceil.$
\end{lemma}
\begin{proof}
Let $L$ be the linear subspace of $\mathbb{F}_2^{q}$ spanned by
the incidence vectors of the sets $\alpha+\beta S_1,$ for $\alpha
\in \mathbb{F}_q$ and $\beta \in S.$ Observe that $L$ is invariant
under the actions of a $1$-transitive permutation group (permuting
the coordinates in accordance with addition in $\mathbb{F}_q$).
This implies that the space $L^{\perp}$ is also invariant under
the actions of the same group. Note that $L^\perp$ has positive
dimension since it contains the incidence vector of the set $S_0.$
The last two observations imply that $L^{\perp}$ has {\it full
support,} i.e., for every $i\in [q]$ there exists a vector $v\in
L^\perp$ such that $v_i\ne 0.$ It is easy to verify that any
linear subspace of $\mathbb{F}_2^q$ that has full support contains
a vector of Hamming weight at least $\lceil q/2 \rceil.$ Let $v\in
L^{\perp}$ be such a vector. Redefining the set $S_0$ to be the
set of nonzero coordinates of $v$ we conclude the proof.
\end{proof}

We now proceed to the core proposition of~\cite{Y_nice} that shows
how sets exhibiting both combinatorial and algebraic niceness
yield locally decodable codes.

\begin{proposition}
\label{Prop:SetsToCodes} Suppose $S\subseteq \mathbb{F}_q^*$ is
$t$ combinatorially nice and $k$ algebraically nice; then for
every positive integer $n$ there exists a code of length
$\exp(n^{1/t})$ that is $(k,\delta,2k\delta)$ locally decodable
for all $\delta>0.$
\end{proposition}
\begin{proof}
Our proof comes in three steps. We specify encoding and local
decoding procedures for our codes and then argue the lower bound
for the probability of correct decoding. We use the notation from
definitions~\ref{Def:CombinatoriallyNice}
and~\ref{Def:AlgebraicallyNice}.

{\it Encoding:} We assume that our message has length $n=\lfloor
cm^t \rfloor$ for some value of $m.$ (Otherwise we pad the message
with zeros. It is easy to see that such padding does not not
affect the asymptotic length of the code.) Our code will be
linear. Therefore it suffices to specify the encoding of unit
vectors $e_1,\ldots,e_n,$ where $e_j$ has length $n$ and a unique
non-zero coordinate $j.$ We define the encoding of $e_j$ to be a
$q^m$ long vector, whose coordinates are labelled by elements of
$\mathbb{F}_q^m.$ For all $w\in \mathbb{F}_q^m$ we set:
\begin{equation}
\label{Eqn:Encoding}
 \mbox{Enc}(e_j)_{w}=\left\{
 \begin{array}{ll}
 1, & \mbox{if $(u_j,w)\in S_0;$} \\
 0, & \mbox{otherwise.} \\
 \end{array}
 \right.
\end{equation}
It is straightforward to verify that we defined a code encoding
$n$ bits to $\exp(n^{1/t})$ bits.

{\it Local decoding:} Given a (possibly corrupted) codeword $y$
and an index $i\in [n],$ the decoding algorithm $\mathcal A$ picks
$w\in \mathbb{F}_q^m,$ such that $(u_i,w)\in S_0$ uniformly at
random, reads $k^\prime\leq k$ coordinates of $y,$ and outputs the
sum:
\begin{equation}
\label{Eqn:Decoding}
 \sum\limits_{\lambda\in S_1} y_{w+\lambda v_i}.
\end{equation}

{\it Probability of correct decoding:} First we argue that
decoding is always correct if $\mathcal A$ picks $w\in
\mathbb{F}_q^m$ such that all bits of $y$ in locations
$\{w+\lambda v_i\}_{\lambda\in S_1}$ are not corrupted. We need to
show that for all $i\in [n],$ $x\in \{0,1\}^n$ and $w\in
\mathbb{F}_q^m,$ such that $(u_i,w)\in S_0$:
\begin{equation}
\label{Eqn:SetToCodeStep0} \sum\limits_{\lambda\in S_1} \left(
\sum\limits_{j=1}^n x_j\ \mbox{Enc}(e_j)\right)_{w+\lambda v_i} =
x_i.
\end{equation}
Note that
\begin{equation}
\label{Eqn:SetToCodeStep1} \sum\limits_{\lambda\in S_1} \left(
\sum\limits_{j=1}^n x_j\ \mbox{Enc}(e_j)\right)_{w+\lambda v_i} =
\sum\limits_{j=1}^n x_j \sum\limits_{\lambda\in S_1}
\mbox{Enc}(e_j)_{w+\lambda v_i}= \sum\limits_{j=1}^n x_j
\sum\limits_{\lambda\in S_1} I\left[ (u_j,w+\lambda v_i)\in S_0
\right],
\end{equation}
where $I[\gamma \in S_0]=1$ if $\gamma\in S_0$ and zero otherwise.
Now note that
\begin{equation}
\label{Eqn:SetToCodeStep2} \sum\limits_{\lambda\in S_1}I\left[
(u_j,w+\lambda v_i)\in S_0 \right] = \sum\limits_{\lambda\in
S_1}I\left[ (u_j, w)+\lambda (u_j, v_i)\in S_0 \right] = \left\{
\begin{array}{ll}
1, & \mbox{if $i=j,$} \\
0, & \mbox{otherwise.}
\end{array}\right.
\end{equation}
The last identity in~(\ref{Eqn:SetToCodeStep2}) for $i=j$ follows
from: $(u_i,v_i)=0,$ $(u_i,w)\in S_0$ and $k^\prime=|S_1|$ is odd.
The last identity for $i\ne j$ follows from $(u_j,v_i)\in S$ and
the algebraic niceness of $S.$ Combining
identities~(\ref{Eqn:SetToCodeStep1})
and~(\ref{Eqn:SetToCodeStep2}) we get~(\ref{Eqn:SetToCodeStep0}).

Now assume that up to $\delta$ fraction of bits of $y$ are
corrupted. Let $T_i$ denote the set of coordinates whose labels
belong to $\left\{w\in \mathbb{F}_q^m\ |\ (u_i,w)\in S_0\right\}.$
Recall that by lemma~\ref{Lemma:S0Large}, $|T_i|\geq q^m/2.$ Thus
at most $2\delta$ fraction of coordinates in $T_i$ contain
corrupted bits. Let $Q_i=\left\{ \left\{w+\lambda v_i
\right\}_{\lambda\in S_1} \ |\ w: (u_i,w)\in S_0 \right\}$ be the
family of $k^\prime$-tuples of coordinates that may be queried by
$\mathcal A.$ $(u_i,v_i)=0$ implies that elements of $Q_i$
uniformly cover the set $T_i.$ Combining the last two observations
we conclude that with probability at least $1-2k\delta$ $\mathcal
A$ picks an uncorrupted $k^\prime$-tuple and outputs the correct
value of $x_i.$
\end{proof}

\smallskip

All locally decodable codes constructed in this paper are obtained
by applying proposition~\ref{Prop:SetsToCodes} to certain nice
sets. Thus all our codes have the same dependence of $\epsilon$
(the probability of the decoding error) on $\delta$ (the fraction
of corrupted bits). In what follows we often ignore these
parameters and consider only the length and query complexity of
codes.

\section{Mersenne numbers with large prime factors yield nice subsets of prime fields}\label{Sec:FactorsToSets}
In what follows let $\langle 2\rangle\subseteq \mathbb{F}_p^*$
denote the multiplicative subgroup of $\mathbb{F}_p^*$ generated
by $2.$ In~\cite{Y_nice} it is shown that for every Mersenne prime
$p=2^t-1$ the set $\langle 2\rangle \subseteq \mathbb{F}_p^*$ is
simultaneously $3$ algebraically nice and $\mathrm{ord}_2(p)$
combinatorially nice. In this section we prove the same conclusion
for a substantially broader class of primes.

\begin{lemma}
\label{Lemma:2GroupIsCombinatoriallyNice} Suppose $p$ is an odd
prime; then $\langle 2\rangle\subseteq\mathbb{F}_p^*$ is
$\mathrm{ord}_2(p)$ combinatorially nice.
\end{lemma}
\begin{proof}
Let $t=\mathrm{ord}_2(p).$ Clearly, $t$ divides $p-1.$ We need to
specify a constant $c>0$ such that for every positive integer $m$
there exist two $n=\lfloor cm^t \rfloor$-sized collections of $m$
long vectors over $\mathbb{F}_p$ satisfying:
\begin{itemize}
\item For all $i\in [n],$ $(u_i,v_i)=0;$

\item For all $i,j\in [n]$ such that $i\ne j,$ $(u_j,v_i)\in
\langle 2\rangle.$
\end{itemize}
First assume that $m$ has the shape
$m=\left(m^\prime-1+(p-1)/t\atop (p-1)/t\right),$ for some integer
$m^\prime \geq p-1.$ In this case~\cite[lemma 13]{Y_nice} gives us
a collection of $n=\left(m^\prime\atop p-1\right)$ vectors with
the right properties. Observe that $n\geq cm^t$ for a constant $c$
that depends only on $p$ and $t.$ Now assume $m$ does not have the
right shape, and let $m_1$ be the largest integer smaller than $m$
that does have it. In order to get vectors of length $m$ we use
vectors of length $m_1$ coming from~\cite[lemma 13]{Y_nice} padded
with zeros. It is not hard to verify such a construction still
gives us $n\geq cm^t$ large families of vectors for a suitably
chosen constant $c.$
\end{proof}

We use the standard notation $\overline{\mathbb{F}}$ to denote the
algebraic closure of the field $\mathbb{F}.$ Also let $C_p
\subseteq \overline{\mathbb{F}}_2^*$ denote the multiplicative
subgroup of $p$-th roots of unity in $\overline{\mathbb{F}}_2$.
The next lemma generalizes~\cite[lemma 14]{Y_nice}.

\begin{lemma}
\label{Lemma:SumsOfRootsYieldAlgebricNiceness} Let $p$ be a prime
and $k$ be odd. Suppose there exist $\zeta_1,\ldots,\zeta_k \in
C_p$ such that
\begin{equation}
\label{Eqn:GenericRootsSumToZero} \zeta_1+\ldots+\zeta_k=0;
\end{equation}
then $\langle 2\rangle\subseteq \mathbb{F}_p^*$ is $k$
algebraically nice.
\end{lemma}
\begin{proof}
In what follows we define the set $S_1 \subseteq \mathbb{F}_p$ and
prove the existence of a set $S_0$ such that that together $S_0$
and $S_1$ yield $k$ algebraic niceness of $\langle 2\rangle.$
Identity~\ref{Eqn:GenericRootsSumToZero} implies that there exists
an odd integer $k^\prime\leq k$ and $k^\prime$ {\it distinct}
$p$-th roots of unity $\zeta_1^\prime,\ldots,\zeta_k^\prime \in
C_p$ such that
\begin{equation}
\label{Eqn:DistinctRootsSumToZero}
\zeta_1^\prime+\ldots+\zeta_{k^\prime}^\prime=0.
\end{equation}
Let $t=\mathrm{ord}_2(p).$ Observe that $C_p\subseteq
\mathbb{F}_{2^t}.$ Let $g$ be a generator of $C_p.$
Identity~(\ref{Eqn:DistinctRootsSumToZero}) yields
$g^{\gamma_1}+\ldots+g^{\gamma_{k^\prime}}=0,$ for some distinct
values of $\{\gamma_i\}_{i\in [k^\prime]}.$ Set
$S_1=\{\gamma_1,\ldots,\gamma_{k^\prime}\}.$

Consider a natural one to one correspondence between subsets
$S^\prime$ of $\mathbb{F}_p$ and polynomials $\phi_{S^\prime}(x)$
in the ring $\mathbb{F}_2[x]/(x^p-1):$
$\phi_{S^\prime}(x)=\sum\limits_{s\in S^\prime}x^s. $ It is easy
to see that for all sets $S^\prime\subseteq \mathbb{F}_p$ and all
$\alpha,\beta\in\mathbb{F}_p,$ such that $\beta\ne 0:$
$$ \phi_{\alpha+\beta S^\prime}(x) = x^\alpha\phi_{S^\prime}(x^\beta).$$
Let $\alpha$ be a variable ranging over $\mathbb{F}_p$ and $\beta$
be a variable ranging over $\langle 2\rangle.$ We are going to
argue the existence of a set $S_0$ that has even intersections
with all sets of the form $\alpha+\beta S_1,$ by showing that all
polynomials $\phi_{\alpha+\beta S_1}$ belong to a certain linear
space $L\in \mathbb{F}_2[x]/(x^p-1)$ of dimension less than $p.$
In this case any nonempty set $T\subseteq \mathbb{F}_p$ such that
$\phi_{T}\in L^{\perp}$ can be used as the set $S_0.$ Let
$\tau(x)=\mathrm{gcd}(x^p-1,\phi_{S_1}(x)).$ Note that $\tau(x)\ne
1$ since $g$ is a common root of $x^p-1$ and $\phi_{S_1}(x).$ Let
$L$ be the space of polynomials in $\mathbb{F}_2[x]/(x^p-1)$ that
are multiples of $\tau(x).$ Clearly, $\dim L=p-\deg \tau.$ Fix
some $\alpha\in \mathbb{F}_p$ and $\beta \in \langle 2 \rangle.$
Let us prove that $\phi_{\alpha+\beta S_1}(x)$ is in $L:$
$$
\phi_{\alpha+\beta S_1}(x)=x^\alpha\phi_{S_1}(x^\beta)=
x^\alpha(\phi_{S_1}(x))^\beta.
$$
The last identity above follows from the fact that for any $f \in
\mathbb{F}_2[x]$ and any integer $i:$ $f(x^{2^i})=(f(x))^{2^i}.$
\end{proof}

In what follows we present sufficient conditions for the existence
of $k$-tuples of $p$-th roots of unity in
$\overline{\mathbb{F}}_2$ that sum to zero. We treat the $k=3$
case separately since in that case we can use a specialized
argument to derive a more explicit conclusion.

\subsection{A sufficient condition for the existence of three $p$-th roots of unity summing to zero}
\begin{lemma}
\label{Lemma:WeilYields3Sums} Let $p$ be an odd prime. Suppose
$\mathrm{ord}_2(p)<(4/3)\log_2 p;$ then there exist three $p$-th
roots of unity in $\overline{\mathbb{F}}_2$ that sum to zero.
\end{lemma}
\begin{proof}
We start with a brief review of some basic concepts of projective
algebraic geometry. Let $\mathbb{F}$ be a field, and $f\in
\mathbb{F}[x,y,z]$ be a homogeneous polynomial. A triple
$(x_0,y_0,z_0)\in \mathbb{F}^3$ is called a zero of $f$ if
$f(x_0,y_0,z_0)=0.$ A zero is called nontrivial if it is different
from the origin. An equation $f=0$ defines a projective plane
curve $\chi_f$. Nontrivial zeros of $f$ considered up to
multiplication by a scalars are called $\mathbb{F}$-rational
points of $\chi_f.$ If $\mathbb{F}$ is a finite field it makes
sense to talk about the number of $\mathbb{F}$-rational points on
a curve.

Let $t=\mathrm{ord}_2(p).$ Note that $C_p\subseteq
\mathbb{F}_{2^t}.$ Consider a projective plane Fermat curve $\chi$
defined by
\begin{equation}
\label{Eqn:FermatCurve} x^{(2^t-1)/p} + y^{(2^t-1)/p} +
z^{(2^t-1)/p} = 0.
\end{equation}
Let us call a point $a$ on $\chi$ trivial if one of the
coordinates of $a$ is zero. Cyclicity of $\mathbb{F}_{2^t}^*$
implies that $\chi$ contains exactly $3(2^t-1)/p$ trivial
$\mathbb{F}_{2^t}$-rational points. Note that every nontrivial
point of $\chi$ yields a triple of elements of $C_p$ that sum to
zero. The classical Weil bound~\cite[p. 330]{LN} provides an
estimate
\begin{equation}
\label{Eqn:WeilBound} \left| N_q - (q+1)\right|\leq
(d-1)(d-2)\sqrt{q}
\end{equation}
for the number $N_q$ of $\mathbb{F}_q$-rational points on an
arbitrary smooth projective plane curve of degree $d.$
(\ref{Eqn:WeilBound}) implies that in case
\begin{equation}
\label{Eqn:WeilImplies}
2^t+1>\left(\frac{2^t-1}{p}-1\right)\left(\frac{2^t-1}{p}-2\right)2^{t/2}+3\frac{2^t-1}{p}
\end{equation}
there exists a nontrivial point on the
curve~(\ref{Eqn:FermatCurve}). Note that~(\ref{Eqn:WeilImplies})
follows from
\begin{equation}
\label{Eqn:WeilImplies1}
2^t+1>\left(\frac{2^t}{p}\right)\left(\frac{2^t}{p}\right)2^{t/2}-
\frac{2^{3t/2+1}}{p}+\frac{3*2^t}{p},
\end{equation}
and~(\ref{Eqn:WeilImplies1}) follows from
$$ 2^t > 2^{2t+t/2}/p^2\quad \mbox{ and }\quad 2^{t/2+1} > 3. $$
Now note that the first inequality above follows from
$t<(4/3)\log_2 p$ and the second follows from~$t>1.$
\end{proof}

Note that the constant $4/3$ in lemma~\ref{Lemma:WeilYields3Sums}
cannot be improved to 2: there are no three elements of
$C_{13264529}$ that sum to zero, even though
$\mathrm{ord}_2(13264529) = 47 < 2 * \log_2 13264529 \approx
47.3.$

\subsection{A sufficient condition for the existence of $k$ $p$-th roots of unity summing to zero}
Our argument in this section comes in three steps. First we
briefly review the notion of (additive) Fourier coefficients of
subsets of $\mathbb{F}_{2^t}.$ Next, we invoke a folklore argument
to show that subsets of $\mathbb{F}_{2^t}$ with appropriately
small nontrivial Fourier coefficients contain $k$-tuples of
elements that sum to zero. Finally, we use a recent result of
Bourgain and Chang~\cite{BourgainChang} (generalizing the
classical estimate for Gauss sums) to argue that (under certain
constraints on $p$) all nontrivial Fourier coefficients of $C_p$
are small.

\medskip

For $x\in \mathbb{F}_{2^t}$ let $Tr(x)=x+x^2+ \ldots +x^{2^{t-1}}$
denote the trace of $x.$ It is not hard to verify that for all
$x,$ $Tr(x)\in \mathbb{F}_2.$ Characters of $\mathbb{F}_{2^t}$ are
homomorphisms from the additive group of $\mathbb{F}_{2^t}$ into
the multiplicative group $\{\pm 1\}.$ There exist $2^t$
characters. We denote characters by $\chi_a,$ where $a$ ranges in
$\mathbb{F}_{2^t},$ and set $\chi_a(x)=(-1)^{Tr(ax)}.$ Let $C(x)$
denote the incidence function of a set $C\subseteq
\mathbb{F}_{2^t}.$ For arbitrary $a\in \mathbb{F}_2^t$ the Fourier
coefficient $\chi_a(C)$ is defined by
$\chi_a(C)=\sum\chi_a(x)C(x),$ where the sum is over all $x\in
\mathbb{F}_{2^t}.$ Fourier coefficient $\chi_0(C)=\left|C\right|$
is called trivial, and other Fourier coefficients are called
nontrivial. In what follows $\sum_{\chi}$ stands for summation
over all $2^t$ characters of $\mathbb{F}_{2^t}.$ We need the
following two standard properties of characters and Fourier
coefficients.
\begin{equation}
\label{Eqn:FourierStandard0} \sum_{\chi} \chi(x) = \left\{
\begin{array}{ll}
2^t, & \mbox{if $x=0;$} \\
0, & \mbox{otherwise.}
\end{array}
\right.
\end{equation}
\begin{equation}
\label{Eqn:FourierStandard1} \sum_{\chi} \chi^2(C) = 2^t |C|.
\end{equation}
The following lemma is a folklore.
\begin{lemma}
\label{Lemma:SmallFourierYeildsShortSums} Let $C\subseteq
\mathbb{F}_{2^t}$ and $k\geq 3$ be a positive integer. Let $F$ be
the largest absolute value of a nontrivial Fourier coefficient of
$C.$ Suppose
\begin{equation}
\label{Eqn:FourierSmall} \frac{F}{|C|}<
\left(\frac{|C|}{2^t}\right)^{1/(k-2)}
\end{equation}
then there exist $k$ elements of $C$ that sum to zero.
\end{lemma}
\begin{proof}
Let $M(C)=\#\left\{\zeta_1,\ldots,\zeta_k\in C\ |\ \zeta_1 +
\ldots + \zeta_k = 0 \right\}.$~(\ref{Eqn:FourierStandard0})
yields
\begin{equation}
\label{Eqn:FourierImpliesSums0}
M(C)=\frac{1}{2^t}\sum\limits_{x_1,\ldots,x_k \in
\mathbb{F}_{2^t}} C(x_1)\ldots C(x_k)\sum\limits_{\chi}
\chi(x_1+\ldots +x_k).
\end{equation}
Note that $\chi(x_1+\ldots +x_k)=\chi(x_1)\ldots \chi(x_k).$
Changing the order of summation in~(\ref{Eqn:FourierImpliesSums0})
we get
\begin{equation}
\label{Eqn:FourierImpliesSums1}
M(C)=\frac{1}{2^t}\sum\limits_{\chi} \sum\limits_{x_1,\ldots,x_k
\in \mathbb{F}_{2^t}} C(x_1)\ldots C(x_k)\chi(x_1)\ldots \chi(x_k)
= \frac{1}{2^t}\sum\limits_{\chi} \chi^k(C).
\end{equation}
Note that
\begin{equation}
\label{Eqn:FourierImpliesSums2} \frac{1}{2^t}\sum\limits_{\chi}
\chi^k(C) = \frac{|C|^k}{2^t} + \frac{1}{2^t}\sum\limits_{\chi\ne
\chi_0}\chi^k(C) \geq \frac{|C|^k}{2^t} -
F^{k-2}\frac{1}{2^t}\sum\limits_{\chi} \chi^2(C) =
\frac{|C|^k}{2^t} - F^{k-2}|C|,
\end{equation}
where the last identity follows from~(\ref{Eqn:FourierStandard1}).
Combining~(\ref{Eqn:FourierImpliesSums1})
and~(\ref{Eqn:FourierImpliesSums2}) we conclude
that~(\ref{Eqn:FourierSmall}) implies $M(C)>0.$
\end{proof}

The following lemma is a special case of~\cite[theorem
1]{BourgainChang}.
\begin{lemma}
\label{Lemma:BourgainChang} Assume that $n\ |\ 2^t-1$ and
satisfies the condition
$$
\mathrm{gcd}\left(n,\frac{2^t-1}{2^{t^\prime}-1}\right)<
2^{t(1-\epsilon)-t^\prime},\quad \mbox{for all \ \ } 1\leq
t^\prime < t,\ t^\prime\ |\ t,
$$
where $\epsilon>0$ is arbitrary and fixed. Then for all $a\in
\mathbb{F}_{2^t}^*$
\begin{equation}
\label{Eqn:BourgainEstimate}
\left| \sum\limits_{x\in
\mathbb{F}_{2^t}}(-1)^{Tr(ax^n)} \right| < c_1 2^{t(1-\delta)},
\end{equation}
where $\delta=\delta(\epsilon)>0$ and $c_1=c_1(\epsilon)$ are
absolute constants.
\end{lemma}

Below is the main result of this section. Recall that $C_p$
denotes the set of $p$-th roots of unity in
$\overline{\mathbb{F}}_2.$

\begin{lemma}
\label{Lemma:kSums} For every $c>0$ there exists an odd integer
$k=k(c)$ such that the following implication holds. If $p$ is an
odd prime and $\mathrm{ord}_2(p)<c\log_2 p$ then some $k$ elements
of $C_p$ sum to zero.
\end{lemma}
\begin{proof}
Note that if there exist $k^\prime$ elements of a set $C\subseteq
\overline{\mathbb{F}}_2$ that sum to zero, where $k^\prime$ is
odd; then there exist $k$ elements of $C$ that sum to zero for
every odd $k\geq k^\prime.$ Also note that the sum of all $p$-th
roots of unity is zero. Therefore given $c$ it suffices to prove
the existence of an odd $k=k(c)$ that works for all {\it
sufficiently large} $p.$ Let $t=\mathrm{ord}_2(p).$ Observe that
$p>2^{t/c}.$ Assume $p$ is sufficiently large so that $t>2c.$ Next
we show that the precondition of lemma~\ref{Lemma:BourgainChang}
holds for $n=(2^t-1)/p$ and $\epsilon=1/(2c).$ Let $t^\prime\ |\
t$ and $1\leq t^\prime<t.$ Clearly
$\mathrm{gcd}(2^{t^\prime}-1,p)=1.$ Therefore
\begin{equation}
\label{Eqn:kSums0}
\mathrm{gcd}\left(\frac{2^t-1}{p},\frac{2^t-1}{2^{t^\prime}-1}
\right) =
\frac{2^t-1}{p(2^{t^\prime}-1)}<\frac{2^{t(1-1/c})}{2^{t^\prime}-1},
\end{equation}
where the inequality follows from $p>2^{t/c}.$ Clearly, $t>2c$
yields $2^{t/(2c)}/2>1.$ Multiplying the right hand side
of~(\ref{Eqn:kSums0}) by $2^{t/(2c)}/2$ and using
$2(2^{t^\prime}-1)>2^{t^\prime}$ we get
\begin{equation}
\label{Eqn:kSums1}
\mathrm{gcd}\left(\frac{2^t-1}{p},\frac{2^t-1}{2^{t^\prime}-1}
\right) <2^{t(1-1/(2c))-t^\prime}.
\end{equation}
Combining~(\ref{Eqn:kSums1}) with lemma~\ref{Lemma:BourgainChang}
we conclude that there exist $\delta>0$ and $c_1$ such that for
all $a\in \mathbb{F}_{2^t}^*$
\begin{equation}
\label{Eqn:kSums2} \left| \sum\limits_{x\in
\mathbb{F}_{2^t}}(-1)^{Tr\left(ax^{(2^t-1)/p}\right)} \right| <
c_1 2^{t(1-\delta)}.
\end{equation}
Observe that $x^{(2^t-1)/p}$ takes every value in $C_p$ exactly
$(2^t-1)/p$ times when $x$ ranges over $\mathbb{F}_{2^t}^*.$
Thus~(\ref{Eqn:kSums2}) implies
\begin{equation}
\label{Eqn:kSums3} (2^t-1)(F/p) < c_1 2^{t(1-\delta)},
\end{equation}
where $F$ denotes that largest nontrivial Fourier coefficient of
$C_p.$~(\ref{Eqn:kSums3}) yields $F/p <(2c_1)2^{-\delta t}.$ Pick
$k\geq 3$ to be the smallest odd integer such that $(1-1/c)/(k-2)
< \delta.$ We now have
\begin{equation}
\label{Eqn:kSums4} \frac{F}{p} < 2^{-\frac{(1-1/c)t}{(k-2)}}
\end{equation}
for all sufficiently large values of $p.$ Combining $p>2^{t/c}$
with~(\ref{Eqn:kSums4}) we get
$$
\frac{F}{\left|
C_p\right|}<\left(\frac{\left|C_p\right|}{2^t}\right)^{1/(k-2)},
$$
and the application of
lemma~\ref{Lemma:SmallFourierYeildsShortSums} concludes the proof.
\end{proof}

\subsection{Summary}\label{SubSec:PositiveSummary}
In this section we summarize our positive results and show that
one does not necessarily need to use Mersenne primes to construct
locally decodable codes via the methods of~\cite{Y_nice}. It
suffices to have Mersenne numbers with polynomially large prime
factors. Recall that $P(m)$ denotes the largest prime factor of an
integer $m.$ Our first theorem gets $3$-query LDCs from Mersenne
numbers $m$ with prime factors larger than $m^{3/4}.$

\begin{theorem}
\label{Theorem:MersenneFactorsYield3LDC} Suppose $P(2^t-1)>2^{0.75
t};$ then for every message length $n$ there exists a three query
locally decodable code of length $\exp(n^{1/t}).$
\end{theorem}
\begin{proof}
Let $P(2^t-1)=p.$ Observe that $p\ |\ 2^t-1$ and $p>2^{0.75 t}$
yield $\mathrm{ord}_2(p) <(4/3)\log_2 p.$ Combining
lemmas~\ref{Lemma:WeilYields3Sums},\ref{Lemma:SumsOfRootsYieldAlgebricNiceness}
and~\ref{Lemma:2GroupIsCombinatoriallyNice} with
proposition~\ref{Prop:SetsToCodes} we obtain the statement of the
theorem.
\end{proof}
As an example application of
theorem~\ref{Theorem:MersenneFactorsYield3LDC} one can observe
that $P(2^{23}-1)=178481>2^{(3/4)*23}\approx 155872$ yields a
family of three query locally decodable codes of length
$\exp(n^{1/23}).$ Theorem~\ref{Theorem:MersenneFactorsYield3LDC}
immediately yields:
\begin{theorem}
\label{Theorem:Short3LDC} Suppose for infinitely many $t$ we have
$P(2^t-1)>2^{0.75 t};$ then for every $\epsilon>0$ there exists a
family of three query locally decodable codes of length
$\exp(n^\epsilon).$
\end{theorem}
The next theorem gets constant query LDCs from Mersenne numbers
$m$ with prime factors larger than $m^{\gamma}$ for every value of
$\gamma.$

\begin{theorem}
\label{Theorem:MersenneFactorsYieldkLDC} For every $\gamma>0$
there exists an odd integer $k=k(\gamma)$ such that the following
implication holds. Suppose $P(2^t-1)>2^{\gamma t};$ then for every
message length $n$ there exists a $k$ query locally decodable code
of length $\exp(n^{1/t}).$
\end{theorem}
\begin{proof}
Let $P(2^t-1)=p.$ Observe that $p\ |\ 2^t-1$ and $p>2^{\gamma t}$
yield $\mathrm{ord}_2(p) <(1/\gamma)\log_2 p.$ Combining
lemmas~\ref{Lemma:kSums},\ref{Lemma:SumsOfRootsYieldAlgebricNiceness}
and~\ref{Lemma:2GroupIsCombinatoriallyNice} with
proposition~\ref{Prop:SetsToCodes} we obtain the statement of the
theorem.
\end{proof}
As an immediate corollary we get:
\begin{theorem}
\label{Theorem:ShortkLDC} Suppose for some $\gamma>0$ and
infinitely many $t$ we have $P(2^t-1)>2^{\gamma t};$ then there is
a fixed $k$ such that for every $\epsilon>0$ there exists a family
of $k$ query locally decodable codes of length $\exp(n^\epsilon).$
\end{theorem}

\section{Nice subsets of finite fields yield Mersenne numbers with large prime factors}\label{Sec:SetsToFactors}
\begin{definition}
\label{Def:Nice} We say that a sequence $\left\{S_i \subseteq
\mathbb{F}_{q_i}^*\right\}_{i\geq 1}$ of subsets of finite fields
is $k$-nice if every $S_i$ is $k$ algebraically nice and $t(i)$
combinatorially nice, for some integer valued monotonically
increasing function $t.$
\end{definition}
The core proposition~\ref{Prop:SetsToCodes} asserts that a subset
$S\subseteq \mathbb{F}_q^*$ that is $k$ algebraically nice and $t$
combinatorially nice yields a family of $k$-query locally
decodable codes of length $\exp(n^{1/t}).$ Clearly, to get
$k$-query LDCs of length $\exp(n^\epsilon)$ for some fixed $k$ and
every $\epsilon>0$ via this proposition, one needs to exhibit a
$k$-nice sequence. In this section we show how the existence of a
$k$-nice sequence implies that infinitely many Mersenne numbers
have large prime factors. Our argument proceeds in two steps.
First we show that a $k$-nice sequence yields an infinite sequence
of primes $\left\{p_i\right\}_{i\geq 1},$ where every $C_{p_i}$
contains a $k$-tuple of elements summing to zero. Next we show
that $C_p$ contains a short additive dependence only if $p$ is a
large factor of a Mersenne number.

\subsection{A nice sequence yields infinitely many primes $p$ with short dependencies between $p$-th roots of unity}
We start with some notation. Consider a a finite field
$\mathbb{F}_q=\mathbb{F}_{p^l},$  where $p$ is prime. Fix a basis
$e_1,\ldots,e_l$ of $\mathbb{F}_q$ over $\mathbb{F}_p.$ In what
follows we often write $(\alpha_1,\ldots,\alpha_l)\in
\mathbb{F}_{p}^l$ to denote $\alpha=\sum_{i=1}^l \alpha_i e_i \in
\mathbb{F}_{q}.$ Let $R$ denote the ring
$\mathbb{F}_2[x_1,\ldots,x_l]/(x_1^p-1,\ldots,x_l^p-1).$ Consider
a natural one to one correspondence between subsets $S_1$ of
$\mathbb{F}_q$ and polynomials $\phi_{S_1}(x_1,\ldots,x_l)\in R.$
$$\phi_{S_1}(x_1,\ldots,x_l)=
  \sum\limits_{(\alpha_1,\ldots,\alpha_l) \in S_1}x_1^{\alpha_1}\ldots x_l^{\alpha_l}.$$
It is easy to see that for all sets $S_1\subseteq \mathbb{F}_q$
and all $\alpha,\beta\in\mathbb{F}_q:$
\begin{equation}
\label{Eqn:GeneralShift} \phi_{(\alpha_1,\ldots,\alpha_l)+\beta
S_1}(x_1,\ldots,x_l) =
   x_1^{\alpha_1}\ldots x_l^{\alpha_l} \phi_{\beta S_1}(x_1,\ldots,x_l).
\end{equation}
Let $\Gamma$ be a family of subsets of $\mathbb{F}_q.$ It is
straightforward to verify that a set $S_0\subseteq \mathbb{F}_q$
has even intersections with every element of $\Gamma$ if and only
if $\phi_{S_0}$ belongs to $L^\perp,$ where $L$ is the linear
subspace of $R$ spanned by $\left\{\phi_{S_1}\right\}_{S_1\in
\Gamma}.$ Combining the last observation with
formula~(\ref{Eqn:GeneralShift}) we conclude that a set
$S\subseteq \mathbb{F}_q^*$ is $k$ algebraically nice if and only
if there exists a set $S_1\subseteq \mathbb{F}_q$ of odd size
$k^\prime\leq k$ such that the ideal generated by polynomials
$\left\{\phi_{\beta S_1}\right\}_{\{\beta\in S\}}$ is a proper
ideal of $R.$ Note that polynomials $\{f_1,\ldots,f_h\}\in R$
generate a proper ideal if an only if polynomials
$\{f_1,\ldots,f_h,x_1^p-1,\ldots,x_l^p-1\}$ generate a proper
ideal in $\mathbb{F}_2[x_1,\ldots,x_l].$ Also note that a family
of polynomials generates a proper ideal in
$\mathbb{F}_2[x_1,\ldots,x_l]$ if and only if it generates a
proper ideal in $\overline{\mathbb{F}}_2[x_1,\ldots,x_l].$ Now an
application of Hilbert's Nullstellensatz~\cite[p.~168]{CLS}
implies that a set $S\subseteq \mathbb{F}_q^*$ is $k$
algebraically nice if and only if there is a set $S_1\subseteq
\mathbb{F}_q$ of odd size $k^\prime\leq k$ such that the
polynomials $\left\{\phi_{\beta S_1}\right\}_{\{\beta\in S\}}$ and
$\left\{x_i^p-1\right\}_{1\leq i\leq l}$ have a common root in
$\overline{\mathbb{F}}_2.$
\begin{lemma}
\label{Lemma:kAlgNiceImplieskSum} Let
$\mathbb{F}_q=\mathbb{F}_{p^l},$  where $p$ is prime. Suppose
$\mathbb{F}_q$ contains a nonempty $k$ algebraically nice subset;
then there exist $\zeta_1,\ldots,\zeta_k \in C_p$ such that
$\zeta_1+ \ldots +\zeta_k=0.$
\end{lemma}
\begin{proof}
Assume $S\subseteq \mathbb{F}_q^*$ is nonempty and $k$
algebraically nice. The discussion above implies that there exists
$S_1\subseteq \mathbb{F}_q$ of odd size $k^\prime\leq k$ such that
all polynomials $\left\{\phi_{\beta S_1}\right\}_{\{\beta\in S\}}$
vanish at some $(\zeta_1,\ldots,\zeta_l)\in C_p^{l}.$ Fix an
arbitrary $\beta_0 \in S,$ and note that $C_p$ is closed under
multiplication. Thus,
\begin{equation}
\label{Eqn:SparseSum} \phi_{\beta_0 S_1}(\zeta_1,\ldots,\zeta_l)=0
\end{equation}
yields $k^\prime$ $p$-th roots of unity that add up to zero. It is
readily seen that one can extend~(\ref{Eqn:SparseSum}) (by adding
an appropriate number of pairs of identical roots) to obtain $k$
$p$-th roots of unity that add up to zero for any odd $k\geq
k^\prime.$
\end{proof}
Note that lemma~\ref{Lemma:kAlgNiceImplieskSum} does not suffice
to prove that a $k$-nice sequence $\left\{S_i \subseteq
\mathbb{F}_{q_i}^*\right\}_{i\geq 1}$ yields infinitely many
primes $p$ with short (nontrivial) additive dependencies in $C_p.$
We need to argue that the set $\left\{\mbox{char}
\mathbb{F}_{q_i}\right\}_{i\geq 1}$ can not be finite. To proceed,
we need some more notation. Recall that $q=p^l$ and $p$ is prime.
For $x\in \mathbb{F}_q$ let $Tr(x)=x+ \ldots + x^{p^{l-1}}\in
\mathbb{F}_p$ denote the (absolute) trace of $x.$ For $\gamma \in
\mathbb{F}_q, c\in \mathbb{F}_p^*$ we call the set
$\pi_{\gamma,c}=\left\{ x\in \mathbb{F}_q\ |\ Tr(\gamma x)=c
\right\}$ a {\it proper affine hyperplane} of $\mathbb{F}_q.$
\begin{lemma}
\label{Lemma:kAlgNiceImpliesCover}  Let
$\mathbb{F}_q=\mathbb{F}_{p^l},$  where $p$ is prime. Suppose
$S\subseteq \mathbb{F}_q^*$ is $k$ algebraically nice; then there
exist $h\leq p^k$ proper affine hyperplanes
$\left\{\pi_{\gamma_i,c_i}\right\}_{1\leq i\leq h}$ of
$\mathbb{F}_q$ such that $S\subseteq
\bigcup\limits_{i=1}^{h}\pi_{\gamma_i,c_i}.$
\end{lemma}
\begin{proof} Discussion preceding lemma~\ref{Lemma:kAlgNiceImplieskSum}
implies that there exists a set
$S_1=\{\sigma_1,\ldots,\sigma_{k^\prime}\}\subseteq \mathbb{F}_q$
of odd size $k^\prime\leq k$ such that all polynomials
$\left\{\phi_{\beta S_1}\right\}_{\{\beta\in S\}}$ vanish at some
$(\zeta_1,\ldots,\zeta_l)\in C_p^{l}.$ Let $\zeta$ be a generator
of $C_p.$ For every $1\leq i\leq l$ pick $\omega_i \in
\mathbb{Z}_p$ such that $\zeta_i=\zeta^{\omega_i}.$ For every
$\beta\in S,$ $\phi_{\beta S_1}(\zeta_1,\ldots,\zeta_l)=0$ yields
\begin{equation}
\label{Eqn:ComplexSum0} \sum\limits_{\mu=(\mu_1,\ldots,\mu_l)\in
\beta S_1} \zeta^{\sum_{i=1}^l \mu_i\omega_i}=0.
\end{equation}
Observe that for fixed values $\left\{\omega_i\right\}_{1\leq
i\leq l}\in \mathbb{Z}_p$ the map $D(\mu)=\sum_{i=1}^l
\mu_i\omega_i$ is a linear map from $\mathbb{F}_q$ to
$\mathbb{F}_p.$ It is not hard to prove that every such map can be
expressed as $D(\mu)=Tr(\delta\mu)$ for an appropriate choice of
$\delta\in \mathbb{F}_q.$ Therefore we can
rewrite~(\ref{Eqn:ComplexSum0}) as
\begin{equation}
\label{Eqn:ComplexSum1} \sum\limits_{\mu \in \beta S_1}
\zeta^{Tr(\delta \mu)} = \sum\limits_{\sigma \in S_1}
\zeta^{Tr(\delta \beta \sigma)} = 0.
\end{equation}
Let $W=\left\{(w_1,\ldots,w_{k^\prime}) \in
\mathbb{Z}_p^{k^\prime}\ \left|\right. \ \zeta^{w_1}+\ldots
+\zeta^{w_{k^\prime}}=0 \right\}$ denote the set of exponents of
$k^\prime$-dependencies between powers of $\zeta.$ Clearly,
$\left|W \right|\leq p^k.$ Identity~(\ref{Eqn:ComplexSum1})
implies that every $\beta\in S$ satisfies
\begin{equation}
\label{Eqn:ComplexSystem} \left\{
\begin{array}{lcl}
Tr((\delta\sigma_1)\beta)         & = & w_1,          \\
\vdots                            &   &               \\
Tr((\delta\sigma_{k^\prime})\beta)& = & w_{k^\prime}; \\
\end{array}
\right.
\end{equation}
for an appropriate choice of $(w_1,\ldots,w_{k^\prime})\in W.$
Note that the all-zeros vector does not lie in $W$ since
$k^\prime$ is odd. Therefore at least one of the identities
in~(\ref{Eqn:ComplexSystem}) has a non-zero right-hand side, and
defines a proper affine hyperplane of $\mathbb{F}_q$. Collecting
one such hyperplane for every element of $W$ we get a family of
$|W|$ proper affine hyperplanes containing every element of $S.$
\end{proof}
Lemma~\ref{Lemma:kAlgNiceImpliesCover} gives us some insight into
the structure of algebraically nice subsets of $\mathbb{F}_q.$ Our
next goal is to develop an insight into the structure of
combinatorially nice subsets. We start by reviewing some relations
between tensor and dot products of vectors. For vectors $u\in
\mathbb{F}_q^m$ and $v\in \mathbb{F}_q^n$ let $u\otimes v\in
\mathbb{F}_q^{mn}$ denote the tensor product of $u$ and $v.$
Coordinates of $u\otimes v$ are labelled by all possible elements
of $[m]\times [n]$ and $(u\otimes v)_{i,j}=u_i v_j.$ Also, let
$u^{\otimes l}$ denote the $l$-the tensor power of $u$ and $u\circ
v$ denote the concatenation of $u$ and $v.$ The following identity
is standard. For any $u,x\in \mathbb{F}_q^m$ and $v,y\in
\mathbb{F}_q^n:$
\begin{equation}
\label{Eqn:2TensorDot} \left(u\otimes v, x\otimes y\right)=
\sum\limits_{i\in [m], j\in [n]}u_i v_j x_i y_j =
\left(\sum\limits_{i\in [m]}u_i x_i\right) \left(\sum\limits_{j\in
[n]}v_j y_j\right) = (u,x)(v,y).
\end{equation}
In what follows we need a generalization of
identity~(\ref{Eqn:2TensorDot}). Let $f(x_1,\ldots,x_h)=\sum_i c_i
x_1^{\alpha_1^i}\ldots x_h^{\alpha_h^i}$ be a polynomial in
$\mathbb{F}_q[x_1,\ldots,x_h].$ Given $f$ we define ${\bar f} \in
\mathbb{F}_q[x_1,\ldots,x_h]$ by ${\bar f} = \sum_i
x_1^{\alpha_1^i}\ldots x_h^{\alpha_h^i},$ i.e., we simply set all
nonzero coefficients of $f$ to $1.$ For vectors $u_1,\ldots,u_h$
in $\mathbb{F}_q^m$ define
\begin{equation}
\label{Eqn:PolyVector} f(u_1,\ldots,u_h)=\circ_{i}\ c_i
u_1^{\otimes \alpha_1^i}\otimes \ldots \otimes u_h^{\otimes
\alpha_h^i}.
\end{equation}
Note that to obtain $f(u_1,\ldots,u_h)$ we
replaced products in $f$ by tensor products and addition by
concatenation. Clearly, $f(u_1,\ldots,u_h)$ is a vector whose
length may be larger than $m.$
\begin{claim}
For every $f\in \mathbb{F}_q[x_1,\ldots,x_h]$ and
$u_1,\ldots,u_h,v_1,\ldots,v_h \in \mathbb{F}_q^m:$
\begin{equation}
\label{Eqn:PolyTensorDot} \left(f(u_1,\ldots,u_h),{\bar
f}(v_1,\ldots,v_h)\right) = f((u_1, v_1), \ldots, (u_h, v_h)).
\end{equation}
\end{claim}
\begin{proof}
Let ${\bf u}=(u_1,\ldots,u_h)$ and ${\bf v}=(v_1,\ldots,v_h).$
Observe that if~(\ref{Eqn:PolyTensorDot}) holds for polynomials
$f_1$ and $f_2$ defined over disjoint sets of monomials then it
also holds for $f=f_1+f_2:$
$$
\begin{array}{c}
\left(f({\bf u}),{\bar f}({\bf v})\right) = \left((f_1+f_2)({\bf
u}),({\bar f}_1+{\bar f_2})({\bf v})\right) = \left(f_1({\bf
u})\circ f_2({\bf u}),{\bar f}_1({\bf v})\circ {\bar f}_2({\bf
v})\right) = \vspace{0.15cm} \\
f_1\left((u_1, v_1), \ldots, (u_h, v_h)\right) + f_2\left((u_1,
v_1), \ldots,
(u_h, v_h)\right) = f\left((u_1, v_1), \ldots, (u_h, v_h)\right).\\
\end{array}
$$
Therefore it suffices to prove~(\ref{Eqn:PolyTensorDot}) for
monomials $f = c x_1^{\alpha_1}\ldots x_h^{\alpha_h}.$ It remains
to notice identity~(\ref{Eqn:PolyTensorDot}) for monomials $f = c
x_1^{\alpha_1}\ldots x_h^{\alpha_h}$ follows immediately from
formula~(\ref{Eqn:2TensorDot}) using induction on
$\sum_{i=1}^h\alpha_i.$
\end{proof}
The next lemma bounds combinatorial niceness of certain subsets of
$\mathbb{F}_q^*.$
\begin{lemma}
\label{Lemma:CoverImliesCombinatorialBound} Let
$\mathbb{F}_q=\mathbb{F}_{p^l},$  where $p$ is prime. Let
$S\subseteq \mathbb{F}_q^*.$  Suppose there exist $h$ proper
affine hyperplanes $\left\{\pi_{\gamma_r,c_r}\right\}_{1\leq r\leq
h}$ of $\mathbb{F}_q$ such that $S\subseteq
\bigcup\limits_{r=1}^{h}\pi_{\gamma_r,c_r};$ then $S$ is at most
$h(p-1)$ combinatorially nice.
\end{lemma}
\begin{proof} Assume $S$ is $t$ combinatorially nice. This implies
that for some $c>0$ and every $m$ there exist two $n=\lfloor
cm^t\rfloor$-sized collections of vectors $\{u_i\}_{i\in [n]}$ and
$\{v_i\}_{i\in [n]}$ in $\mathbb{F}_q^m,$ such that:
\begin{itemize}
\item For all $i\in [n],$ $(u_i,v_i)=0;$

\item For all $i,j\in [n]$ such that $i\ne j,$ $(u_j,v_i)\in S.$
\end{itemize}
For a vector $u\in \mathbb{F}_q^m$ and integer $e$ let $u^e$
denote a vector resulting from raising every coordinate of $u$ to
the power~$e.$ For every $i\in [n]$ and $r\in [h]$ define vectors
$u_i^{(r)}$ and $v_i^{(r)}$ in $\mathbb{F}_q^{ml}$ by
\begin{equation}
\label{Eqn:Vectors0} u_i^{(r)}=(\gamma_r u_i)\circ (\gamma_r
u_i)^{p}\circ \ldots \circ (\gamma_r u_i)^{p^{l-1}}\quad
\mbox{and} \quad v_i^{(r)}=v_i\circ v_i^{p}\circ \ldots \circ
v_i^{p^{l-1}}.
\end{equation}
Note that for every $r_1,r_2 \in [h],\ v_i^{(r_1)}=v_i^{(r_2)}.$
It is straightforward to verify that for every $i,j\in [n]$ and
$r\in [h]:$
\begin{equation}
\label{Eqn:Vectors1}
\left(u_j^{(r)},v_i^{(r)}\right)=Tr(\gamma_r(u_j,v_i)).
\end{equation}
Combining~(\ref{Eqn:Vectors1}) with the fact that $S$ is covered
by proper affine hyperplanes $\pi_{\gamma_i,c_i}$ we conclude that
\begin{itemize}
\item For all $i\in [n]$ and $r\in [h],$
$\left(u_i^{(r)},v_i^{(r)}\right)=0;$

\item For all $i,j\in [n]$ such that $i\ne j,$ there exists $r\in
[h]$ such that $\left(u_j^{(r)},v_i^{(r)}\right)\in
\mathbb{F}_p^*.$
\end{itemize}
Pick $g(x_1,\ldots,x_h)\in \mathbb{F}_p[x_1,\ldots,x_h]$ to be a
homogeneous degree $h$ polynomial such that for ${\bf
a}=(a_1,\ldots,a_h)\in \mathbb{F}_p^h:$ $g({\bf a})=0$ if and only
if ${\bf a}$ is the all-zeros vector. The existence of such a
polynomial $g$ follows from~\cite[Example 6.7]{LN}. Set
$f=g^{p-1}.$ Note that for ${\bf a}\in \mathbb{F}_p^h:\ f(a)=0$ if
${\bf a}$ is the all-zeros vector, and $f(a)=1$ otherwise. For all
$i\in [n]$ define
\begin{equation}
\label{Eqn:Vectors3} u_i^\prime =
f\left(u_i^{(1)},\ldots,u_i^{(h)}\right)\circ (1) \quad
\mbox{and}\quad v_i^\prime = {\bar
f}\left(v_i^{(1)},\ldots,v_i^{(h)}\right)\circ (-1).
\end{equation}
Note that $f$ and ${\bar f}$ are homogeneous degree $(p-1)h$
polynomials in $h$ variables. Therefore~(\ref{Eqn:PolyVector})
implies that for all $i$ vectors $u_i^\prime$ and $v_i^\prime$
have length $m^\prime\leq h^{(p-1)h}(ml)^{(p-1)h}.$ Combining
identities~(\ref{Eqn:Vectors3}) and~(\ref{Eqn:PolyTensorDot}) and
using the properties of dot products between vectors
$\left\{u_i^{(r)}\right\}$ and $\left\{v_i^{(r)}\right\}$
discussed above we conclude that for every $m$ there exist two
$n=\lfloor cm^t\rfloor$-sized collections of vectors
$\{u_i^\prime\}_{i\in [n]}$ and $\{v_i^\prime\}_{i\in [n]}$ in
$\mathbb{F}_q^{m^\prime},$ such that:
\begin{itemize}
\item For all $i\in [n],$ $(u_i^\prime,v_i^\prime)=-1;$

\item For all $i,j\in [n]$ such that $i\ne j,$ $(u_j,v_i)=0.$
\end{itemize}
It remains to notice that a family of vectors with such properties
exists only if $n\leq m^\prime,$ i.e., $\lfloor cm^t\rfloor\leq
h^{(p-1)h}(ml)^{(p-1)h}.$ Given that we can pick $m$ to be
arbitrarily large, this implies that $t\leq (p-1)h.$
\end{proof}
The next lemma presents the main result of this section.
\begin{lemma}
\label{Lemma:NiceSequenceYieldsShortSums} Let $k$ be an odd
integer. Suppose there exists a $k$-nice sequence; then for
infinitely many primes $p$ some $k$ of elements of $C_p$ add up to
zero.
\end{lemma}
\begin{proof} Assume $\left\{S_i \subseteq
\mathbb{F}_{q_i}^*\right\}_{i\geq 1}$ is $k$-nice. Let $p$ be a
fixed prime. Combining lemmas~\ref{Lemma:kAlgNiceImpliesCover}
and~\ref{Lemma:CoverImliesCombinatorialBound} we conclude that
every $k$ algebraically nice subset $S\subseteq
\mathbb{F}_{p^l}^*$ is at most $(p-1)p^k$ combinatorially nice.
Note that our bound on combinatorial niceness is independent of
$l.$ Therefore there are only finitely many extensions of the
field $\mathbb{F}_p$ in the sequence
$\left\{\mathbb{F}_{q_i}\right\}_{i\geq 1},$ and the set
$\mathbb{P}=\left\{\mbox{char} \mathbb{F}_{q_i}\right\}_{i\geq 1}$
is infinite. It remains to notice that according to
lemma~\ref{Lemma:kAlgNiceImplieskSum} for every $p\in \mathbb{P}$
there exist $k$ elements of $C_p$ that add up to zero.
\end{proof}

In what follows we present necessary conditions for the existence
of $k$-tuples of $p$-th roots of unity in
$\overline{\mathbb{F}}_2$ that sum to zero. We treat the $k=3$
case separately since in that case we can use a specialized
argument to derive a slightly stronger conclusion.

\subsection{A necessary condition for the existence of $k$ $p$-th roots of unity summing to zero}
\begin{lemma}
\label{Lemma:kSums} Let $k\geq 3$ be odd and $p$ be a prime.
Suppose there exist $\zeta_1,\ldots,\zeta_k \in C_p$ such that
$\sum_{i=1}^k \zeta_i = 0;$ then
\begin{equation}
\label{Eqn:SumImpliesOrder} \mathrm{ord}_2(p) \leq 2p^{1-1/(k-1)}.
\end{equation}
\end{lemma}
\begin{proof}
Let $t=\mathrm{ord}_2(p).$ Note that $C_p\subseteq
\mathbb{F}_{2^t}.$ Note also that all elements of $C_p$ other than
the multiplicative identity are proper elements of
$\mathbb{F}_{2^t}.$ Therefore for every $\zeta\in C_p$ where
$\zeta\ne 1$ and every $f(x)\in \mathbb{F}_2[x]$ such that $\deg
f\leq t-1$ we have: $f(\zeta)\ne 0.$

By multiplying $\sum_{i=1}^k \zeta_i=0$ through by
$\zeta_{k}^{-1},$ we may reduce to the case $\zeta_{k}=1.$ Let
$\zeta$ be the generator of $C_p.$ For every $i\in [k-1]$ pick
$w_i \in \mathbb{Z}_p$ such that $\zeta_i=\zeta^{w_i}.$  We now
have $\sum_{i=1}^{k-1}\zeta^{w_i}+1=0.$ Set $h=\lfloor (t-1)/2
\rfloor.$ Consider the $(k-1)$-tuples:
\begin{equation}
\label{Eqn:Tuples} (mw_1+i_1,\ldots,mw_{k-1}+i_{k-1})\in
\mathbb{Z}_p^{k-1},\ \mbox{ for }\ m\in \mathbb{Z}_p \ \mbox{ and
}\ i_1,\ldots,i_{k-1}\in [0,h].
\end{equation}
Suppose two of these coincide, say
$$
(mw_1+i_1,\ldots,mw_{k-1}+i_{k-1}) = (m^\prime
w_1+i_1^\prime,\ldots,m^\prime w_{k-1}+i_{k-1}^\prime),
$$
with $(m,i_1,\ldots,i_{k-1})\ne
(m^\prime,i^\prime_1,\ldots,i^\prime_{k-1}).$ Set $n=m-m^\prime$
and $j_l=i_l^\prime-i_l$ for $l\in [k-1].$ We now have
$$
(n w_1,\ldots,n w_{k-1})=(j_1,\ldots,j_l)
$$
with $-h \leq j_1,\ldots,j_{k-1}\leq h.$ Observe that $n\ne 0,$
and thus it has a multiplicative inverse $g\in \mathbb{Z}_p.$
Consider a polynomial
$$
P(z)=z^{j_1+h}+ \ldots +z^{j_{k-1}+h}+z^h \in \mathbb{F}_2[z].
$$
Note that $\deg P\leq 2h\leq t-1.$ Note also that $P(1)=1$ and
$P(\zeta^g)=0.$ The latter identity contradicts the fact that
$\zeta^g$ is a proper element of $\mathbb{F}_{2^t}.$ This
contradiction implies that all $(k-1)$-tuples
in~(\ref{Eqn:Tuples}) are distinct. This yields
$$
p^{k-1} \geq p \left( \frac{t}{2} \right)^{k-1},
$$
which is equivalent to~(\ref{Eqn:SumImpliesOrder}).
\end{proof}

\subsection{A necessary condition for the existence of three $p$-th roots of unity summing to zero}
In this section we slightly strengthen lemma~\ref{Lemma:kSums} in
the special case when $k=3.$ Our argument is loosely inspired by
the Agrawal-Kayal-Saxena deterministic primality test~\cite{AKS}.
\begin{lemma}
\label{Lemma:3Sums} Let $p$ be a prime. Suppose there exist
$\zeta_1,\zeta_2,\zeta_3 \in C_p$ that sum up to zero; then
\begin{equation}
\label{Eqn:3SumImpliesOrder} \mathrm{ord}_2(p) \leq
\left((4/3)p\right)^{1/2}.
\end{equation}
\end{lemma}
\begin{proof}
Let $t=\mathrm{ord}_2(p).$ Note that $C_p\subseteq
\mathbb{F}_{2^t}.$ Note also that all elements of $C_p$ other than
the multiplicative identity are proper elements of
$\mathbb{F}_{2^t}.$ Therefore for every $\zeta\in C_p$ where
$\zeta\ne 1$ and every $f(x)\in \mathbb{F}_2[x]$ such that $\deg
f\leq t-1$ we have: $f(\zeta)\ne 0.$

Observe that $\zeta_1+\zeta_2+\zeta_3=0$ implies
$\zeta_1\zeta_2^{-1}+1=\zeta_3\zeta_2^{-1}.$ This yields
$\left(\zeta_1\zeta_2^{-1}+1\right)^p=1.$ Put
$\zeta=\zeta_1\zeta_2^{-1}.$ Note that $\zeta\ne 1$ and
$\zeta,1+\zeta \in C_p.$ Consider the products $\pi_{i,j} =
\zeta^i (1+\zeta)^j \in C_p$ for $0 \leq i,j \leq t-1.$ Note that
$\pi_{i,j}, \pi_{k,l}$ cannot be the same if $i \geq k$ and $l
\geq j$, as then
\[
\zeta^{i-k} - (1+\zeta)^{l-j} = 0,
\]
but the left side has degree less than $t$. In other words, if
$\pi_{i,j}=\pi_{k,l}$ and $(i,j) \neq (k,l)$, then the pairs
$(i,j)$ and $(k,l)$ are comparable under termwise comparison. In
particular, either $(k,l) = (i+a,j+b)$ or $(i,j) = (k+a,l+b)$ for
some pair $(a,b)$ with $\pi_{a,b}=1$.

We next check that there cannot be two distinct nonzero pairs
$(a,b), (a^\prime,b^\prime)$ with
$\pi_{a,b}=\pi_{a^\prime,b^\prime}=1$. As above, these pairs must
be comparable; we may assume without loss of generality that $a
\leq a^\prime, b \leq b^\prime$. The equations $\pi_{a,b}=1$ and
$\pi_{a^\prime-a,b^\prime-b}=1$ force $a+b \geq t$ and
$(a^\prime-a) + (b^\prime-b) \geq t$, so $a^\prime+b^\prime \geq
2t$. But $a^\prime,b^\prime \leq t-1$, contradiction.

If there is no nonzero pair $(a,b)$ with $0 \leq a,b \leq t-1$ and
$\pi_{a,b}=1$, then all $\pi_{i,j}$ are distinct, so $p \geq t^2$.
Otherwise, as above, the pair $(a,b)$ is unique, and the pairs
$(i,j)$ with $0 \leq i,j \leq t-1$ and $(i,j) \not\geq (a,b)$ are
pairwise distinct. The number of pairs excluded by the condition
$(i,j) \not\geq (a,b)$ is $(t-a)(t-b)$; since $a+b \geq t$,
$(t-a)(t-b) \leq t^2/4$. Hence $p \geq t^2 - t^2/4 = 3t^2/4$ as
desired.
\end{proof}

While the necessary condition given by lemma~\ref{Lemma:3Sums} is
quite far away from the sufficient condition given by
lemma~\ref{Lemma:WeilYields3Sums}, it nonetheless suffices for
checking that for most primes $p$, there do not exist three $p$-th
roots of unity summing to zero. For instance, among the 664578 odd
primes $p \leq 10^8$, all but 550 are ruled out by Lemma 23.
(There is an easy argument that $t$ must be odd if $p > 3$; this
cuts the list down to 273 primes.) Each remaining $p$ can be
tested by computing $\mathrm{gcd}(x^p + 1, (x+1)^p + 1)$; the only
examples we found that did not satisfy the condition of
lemma~\ref{Lemma:WeilYields3Sums} were $(p,t) = (73,9), (262657,
27), (599479, 33), (121369, 39).$

\subsection{Summary}\label{SubSec:NegativeSummary}
In the beginning of this section~\ref{Sec:SetsToFactors} we argued
that in order to use the method of~\cite{Y_nice}, (i.e.,
proposition~\ref{Prop:SetsToCodes}) to obtain $k$-query locally
decodable codes of length $\exp(n^\epsilon)$ for some fixed $k$
and all $\epsilon>0,$ one needs to exhibit a $k$-nice sequence of
subsets of finite fields. In what follows we use technical results
of the previous subsections to show that the existence of a
$k$-nice sequence implies that infinitely many Mersenne numbers
have large prime factors.
\begin{theorem}
\label{Theorem:kNiceSequenceYieldsFactors} Let $k$ be odd. Suppose
there exists a $k$-nice sequence of subsets of finite fields; then
for infinitely many values of $t$ we have
\begin{equation}
\label{Eqn:PLowerk} P(2^t-1)\geq (t/2)^{1+1/(k-2)}.
\end{equation}
\end{theorem}
\begin{proof} Using lemmas~\ref{Lemma:NiceSequenceYieldsShortSums}
and~\ref{Lemma:kSums} we conclude that a $k$-nice sequence yields
infinitely many primes $p$ such that $\mathrm{ord}_2(p)\leq
2p^{1-1/(k-1)}.$ Let $p$ be such a prime and
$t=\mathrm{ord}_2(p).$ Then $P(2^t-1)\geq (t/2)^{1+1/(k-2)}.$
\end{proof}
A combination of lemmas~\ref{Lemma:NiceSequenceYieldsShortSums}
and~\ref{Lemma:3Sums} yields a slightly stronger bound for the
special case of $3$-nice sequences.
\begin{theorem}
\label{Theorem:3NiceSequenceYieldsFactors} Suppose there exists a
$3$-nice sequence of subsets; then for infinitely many values of
$t$ we have
\begin{equation}
\label{Eqn:PLower3} P(2^t-1)\geq (3/4)t^{2}.
\end{equation}
\end{theorem}

We would like to remind the reader that although the lower bounds
for $P(2^t-1)$ given by~(\ref{Eqn:PLowerk})
and~(\ref{Eqn:PLower3}) are extremely weak light of the widely
accepted conjecture saying that the number of Mersenne primes is
infinite, they are substantially stronger than what is currently
known unconditionally~(\ref{Eqn:SteImpliedPrimeFactor}).

\section{Conclusion}\label{Sec:Conclusions}

Recently~\cite{Y_nice} came up with a novel technique for
constructing locally decodable codes and obtained vast
improvements upon the earlier work. The construction proceeds in
two steps. First~\cite{Y_nice} shows that if there exist subsets
of finite fields with certain 'nice' properties then there exist
good codes. Next~\cite{Y_nice} constructs nice subsets of prime
fields $\mathbb{F}_p$ for Mersenne primes $p$.

In this paper we have undertaken an in-depth study of nice subsets
of general finite fields. We have shown that constructing nice
subsets is closely related to proving lower bounds on the size of
largest prime factors of Mersenne numbers. Specifically we
extended the constructions of~\cite{Y_nice} to obtain nice subsets
of prime fields $\mathbb{F}_p$ for primes $p$ that are large
factors of Mersenne numbers. This implies that strong lower bounds
for size of the largest prime factors of Mersenne numbers yield
better locally decodable codes. Conversely, we argued that if one
can obtain codes of subexponential length and constant query
complexity through nice subsets of finite fields then infinitely
many Mersenne numbers have prime factors larger than known
currently.

\section*{Acknowledgements}
Kiran Kedlaya's research is supported by NSF CAREER grant
DMS-0545904 and by the Sloan Research Fellowship. Sergey Yekhanin
would like to thank Swastik Kopparty for providing the
reference~\cite{BourgainChang} and outlining the proof of
lemma~\ref{Lemma:SmallFourierYeildsShortSums}. He would also like
to thank Henryk Iwaniec, Carl Pomerance and Peter Sarnak for their
feedback regarding the number theory problems discussed in this
paper.

\end{document}